\documentclass[11pt]{article}

\usepackage{authblk}
\usepackage{amsfonts, amsmath, amssymb, amsthm}
\usepackage[T1]{fontenc} %
\usepackage[latin1]{inputenc} %
\usepackage[dvips]{graphicx}
\usepackage[normalsize]{subfigure}
\usepackage{color, fancybox, calc}
\usepackage{latexsym,amssymb,amsmath}
\usepackage{fullpage}
\usepackage{array,url, appendix}
\usepackage{refcount}
\usepackage{algorithm}
 \usepackage{algorithmic,tikz}
 
\newtheorem{theorem}{Theorem}
\newtheorem{lemma}[theorem]{Lemma}

\begin{document}
\title{The Erd\H{o}s-Hajnal Conjecture for Long Holes and Anti-holes}
\author[1]{Marthe \textsc{Bonamy}\thanks{Email: \texttt{bonamy@lirmm.fr} Partially supported by the ANR Project \textsc{EGOS} 
    under \textsc{Contract ANR-12-JS02-002-01}}}
\author[1]{Nicolas \textsc{Bousquet}\thanks{Email: \texttt{bousquet@lirmm.fr}}}
\affil[1]{Universit\'e Montpellier 2 - CNRS, LIRMM, 161 rue Ada, 34392 Montpellier, France}
\author[2]{St\'ephan \textsc{Thomass\'e}\thanks{Email: \texttt{stephan.thomasse@ens-lyon.fr} Partially supported by the ANR Project \textsc{Stint} under \textsc{Contract ANR-13-BS02-0007}
}}
\affil[2]{LIP, UMR 5668 ENS Lyon - CNRS - UCBL - INRIA, Universit\'e de Lyon, France.}
\date{}
\renewcommand\Authands{ and }
\maketitle

\begin{abstract} 
Erd\H{o}s and Hajnal conjectured that, for every graph $H$, there exists a constant $c_H$ such that every graph $G$ on $n$ vertices which does not contain any induced copy of $H$ has a clique or a stable set of size $n^{c_H}$. We prove that for every $k$, there exists $c_k>0$ such that 
every graph $G$ on $n$ vertices not inducing a cycle of length at least $k$ 
nor its complement contains a clique or a stable set of size $n^{c_k}$.
\end{abstract}

\section{Introduction}

Let $G=(V,E)$ be a graph. In the following $n$ will denote the size of $V(G)$.
A class $\mathcal C$ of graphs (in this paper, a graph class is closed under
induced subgraphs) is said to satisfy the \emph{(weak) Erd\H{o}s-Hajnal property} if there 
exists some constant $c>0$ such that every graph in $\mathcal C$ on $n$ vertices contains a clique or a stable set of size $n^c$. The Erd\H{o}s-Hajnal conjecture~\cite{Erdo89} asserts that every strict class of graphs satisfies the Erd\H{o}s-Hajnal property. Alon, Pach and Solymosi proved in~\cite{Alon01} that the Erd\H{o}s-Hajnal conjecture is preserved by modules (a \emph{module} is a subset $V_1$ of vertices such that for every $x,y \in V_1$, we have $N(x) \setminus V_1 = N(y) \setminus V_1$): in other words, it suffices to prove that the class of graphs which do not contain any copy of $H$ satisfy the Erd\H{o}s-Hajnal property, for every prime graph $H$ (a \emph{prime graph} is a graph with only trivial modules). The conjecture is satisfied for every prime graph of size at most $4$. For $k=5$, the conjecture is satisfied for bulls~\cite{ChudnovskyS08} but remains open for two prime graphs: the path and the cycle on $5$ vertices. Recently, a new approach for tackling this conjecture has been introduced: forbidding both a graph and its complement. This approach provides a large amount of results for paths (see~\cite{BousquetLT13,Chud12,ChudMP,ChudSey13} for instance). In particular Bousquet, Lagoutte and Thomass\'e proved that, for every $k$, the class of graphs with no $P_k$ nor its complement satisfies the Erd\H{o}s-Hajnal property. A survey of Chudnovsky~\cite{Chud13} details all the known results about this conjecture.

In this paper, we explore the case where long holes and their complements are forbidden (a \emph{hole} is an induced cycle of length at least $4$). A long outstanding open problem due to Gyarf\'as~\cite{Gyar87} asks if, for every integer $k$, the class of graphs with no hole of length at least $k$ is $\chi$-bounded. Equivalently, can the chromatic number of a graph with no hole of length at least $k$ be bounded by a function of its maximal clique and $k$? This question is widely open, since it is even open to determine if a triangle-free graph with no long hole contains a stable set of linear size. Several links exist between Erd\H{o}s-Hajnal property and $\chi$-boundedness. In particular, if the chromatic number of any graph of a class $\mathcal{C}$ is bounded by a polynomial of the maximum clique, the  Erd\H{o}s-Hajnal property holds. Here, we prove that graphs which contain neither a hole of length at least $k$ nor its complement have the Erd\H{o}s-Hajnal property.

\begin{theorem}\label{thm:ehhole}
For every integer $k$, the class of graphs with no holes or anti-holes of length at least $k$ has the Erd\H{o}s-Hajnal property.
\end{theorem}
The remaining of this paper is devoted to a proof of Theorem~\ref{thm:ehhole}.

\section{Dominating tree}

Let $G=(V,E)$ be a connected graph. The \emph{neighborhood} of a set of vertices $X$, denoted by $N(X)$, 
is the set of vertices at distance one from $X$. The \emph{closed neighborhood} of $X$ is $\overline{N}(X)=N(X) \cup X$. By abuse of notation we drop the braces when $X$ contains a single element.
We select a root $r$ in $V$. 
Let $X$ be a set of vertices in $G$. A vertex $y$ in $N(X)$ is {\it active}
for $X$ if it has a neighbor which is not in $\overline{N}(X)$.
The following algorithm returns a subtree $T$ of $G$ rooted at 
$r$ which dominates $G$, i.e. such that every vertex 
of $G$ is at distance at most 1 of $T$.\\

\begin{algorithm}[H] \caption{Find a dominating tree.}\label{algo1}

\begin{algorithmic} 
\REQUIRE A graph $G$, a root $r$.
\ENSURE A dominating tree $T$ rooted at $r$.

\STATE 1 STACK:= $\{r\}$; $T:=\{r\}$

2 While STACK is non empty do

3 \hspace{0.6cm}$x:=$top(STACK)

4 \hspace{0.6cm}If there exists $y$ active for $T$ such that the only neighbor of $y$ in STACK is $x$ 

5 \hspace{1.2cm} Add $y$ to the top of STACK 

6 \hspace{1.2cm} Set $x$ as the father of $y$ in the tree $T$

7 \hspace{0.6cm}Else remove $x$ from STACK and keep track of the deletion order

8 Return $T$ and the deletion order 
\end{algorithmic}
\end{algorithm}

First note that there are two orders on $T$ inherited from Algorithm~\ref{algo1}: the descendance order (well defined since $T$ is rooted) and the deletion order. Observe that STACK is at every step of the algorithm an induced path originated from the root. Indeed, when a vertex is added to STACK its only neighbor in STACK is the top vertex. Conversely, every path in $T$ containing the root $r$ corresponds to the set of vertices in STACK at a given step of the algorithm.
Note also that $T$ dominates $G$ at the end of the algorithm (since $G$ is connected).

\begin{figure}[!h]
\centering
\begin{tikzpicture}[scale=1.3]
\tikzstyle{whitenode}=[draw,circle,fill=white,minimum size=8pt,inner sep=0pt]
\tikzstyle{blacknode}=[draw,circle,fill=black,minimum size=6pt,inner sep=0pt]
\tikzstyle{invisible}=[draw=white,circle,fill=white,minimum size=6pt,inner sep=0pt]
\tikzstyle{tnode}=[draw,ellipse,fill=white,minimum size=8pt,inner sep=0pt]
\tikzstyle{texte} =[fill=white, text=black]

\draw (0,0) node[blacknode] (r) [label=90:$r$] {}
-- ++(-90-45:1.414cm) node[blacknode] (x1) [label=90:$u$] {};

\draw (r)
-- ++(-90:1cm) node[blacknode] (x2) {};
\draw (r)
-- ++(-45:1.414cm) node[blacknode] (x3)  [label=90:$w$] {};

\draw (x1) edge node  {} (x2);

\draw (x1)
-- ++(-90-45:1.414cm) node[blacknode] (y1) {};
\draw (x1)
-- ++(-90:1cm) node[blacknode] (y2) {}
-- ++(-90:1cm) node[blacknode] (z1) {};
\draw (x1)
-- ++(-45:1.414cm) node[blacknode] (y3) [label=90:$v$] {};

\draw (z1) edge node  {} (y1);
\draw (z1) edge node  {} (y3);

\draw (x3)
-- ++(-45:1.414cm) node[blacknode] (y4) {};

\draw (y3) edge node  {} (x3);

\draw[line width=2pt] (r) edge node  {} (x1);
\draw[line width=2pt] (x1) edge node  {} (y1);
\draw[line width=2pt] (r) edge node  {} (x3);

\end{tikzpicture}
\caption{An illustration of Algorithm 1. The tree $T$ is represented with thick edges. The vertices $u$ and $w$ are unrelated, and we have $r(v)=u$.}
\label{fig:algo}
\end{figure}
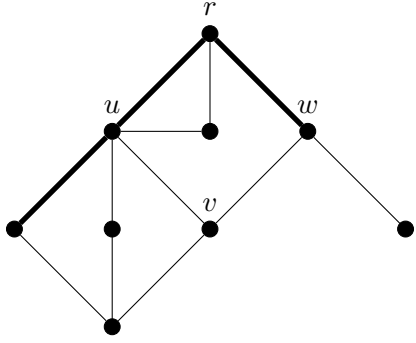

For every set $N$ of nodes of $T$, we denote by $m(N)$
the minimal nodes of $N$ with respect to the descendance order.
In other words, $m(N)$ is a minimum subset of $N$
such that every node of $N$ is a descendant in $T$ of some
(unique) node of $m(N)$. The {\it root} $R(N)$ of $N$ is the minimum element of $m(N)$ with respect to the deletion order.
Equivalently, if we picture the tree $T$ as being built from top to bottom and left to right (when a new vertex is added in STACK, it is drawn just beneath its only neighbor $x$ in $T$, and to the right of any other neighbor of $x$), then $R(N)$ is the leftmost top element of $N$ in $T$.

To any vertex $v$ in $G$, we associate a node $r(v)$ in $T$
by letting $r(v):=R(\overline {N}(v)\cap T)$ (see Figure~\ref{fig:algo} for an illustration). Note that when $v$ is a node
of $T\setminus r$, the vertex $r(v)$ is the father of $v$. Observe also that 
$r(r)=r$, and that $r(x)$ is well defined since $T$ is a dominating set.
Two nodes of the tree are {\it related} if one is a descendant 
of the other, otherwise they are {\it unrelated}. By extension, two sets $A$ and $B$ of nodes are {\it unrelated} if every $a\in A$ and $b\in B$ are unrelated. 

\begin{lemma}\label{keylemma}
If $xy$ is an edge of $G$, then $r(x)$ and $r(y)$ are related.
\end{lemma}

\begin{proof} 
Without loss of generality, we can assume that $r(x)$ was first added in STACK. If $r(x)=x$ then $r(x)$ is $r$ and $r(x)$ and $r(y)$ are related, so we can assume that $r(x) \neq x$. If $r(y)$ is added in STACK before $r(x)$ is deleted, then they are related (all the nodes added in STACK between the addition and the deletion of $x$ are descendants of $x$ since $x$ is on paths from $r$ to these nodes). If $r(y)$ has not been added in STACK when $r(x)$ is deleted, then $x$ is an active neighbor of $r(x)$, as $y$ is a neighbor of $x$ and has no neighbor in the current $T$ (otherwise $r(y)$ would already have appeared in $T$). Consequently, $x$ is added to STACK, a contradiction to the fact that $y$ has no neighbor in $T$ until $r(y)$ is added to STACK.
\end{proof}

\section{Dominating path}

We think that the following result holds for $\frac{1}{3}$, but in our case this easy proof for $\frac{1}{4}$ suffices.

\begin{lemma} \label{weighttree}
Let $T$ be a tree and $w$ be a nonnegative weight function defined on the nodes of
$T$. We assume that $w(T)=1$. Then there is a path $P$ from the root 
with weight at least $\frac{1}{4}$ (i.e. $w(V(P))\geq \frac{1}{4}$) or two unrelated sets 
$A$ and $B$ both with weight at least $\frac{1}{4}$.
\end{lemma}

\begin{proof} 
We grow a path $P$ from the root $r$ by inductively 
adding to the endvertex $x$ of $P$ the root of the heaviest subtree among the sons of $x$. 
If the path $P$ has weight at least $\frac{1}{4}$, we are done. 

Otherwise if there exists a connected component $A$ (i.e. a subtree) of $T\setminus V(P)$ which has
weight at least $\frac{1}{4}$, we conclude as follow: the father $x$ of
the root $z$ of $A$ belongs to $P$. In particular, we did not choose
$z$ to extend $P$ from $x$. Thus $x$ has another son which is the root 
of a subtree $B$ of weight at least $\frac{1}{4}$. We now have our $A$ and $B$ (they are unrelated 
since the two subtrees do not intersect).

In the last case, every component of $T\setminus V(P)$ has weight less that $\frac{1}{4}$, and the 
total weight of $T\setminus V(P)$ is at least $\frac{3}{4}$. We can then group
the components of $T\setminus V(P)$ into two unrelated sets of weight 
at least $\frac{1}{4}$. 
Indeed we iteratively add components until their union $A$ weighs at least $\frac{1}{4}$. Considering that each component weighs less than $\frac{1}{4}$, the weight of $A$ is less than $\frac{1}{2}$.
Since the vertices of the path $P$ have weight less than $\frac{1}{4}$, the remaining connected components have weight at least $\frac{1}{4}$, which gives $B$.
\end{proof}

Let us say that a graph $G$ on $n$ vertices is {\it sparse} if either its maximum
degree is at most $\varepsilon n$ for some small $\varepsilon$,
or it has no triangle. All the sparse graphs we will consider here satisfy the first hypothesis, but we choose that looser notion of sparsity in order to make Lemma~\ref{sparsedominatingpath} more general. 

In a graph $G$, a \emph{complete $\ell$-bipartite 
graph} is a pair of disjoint subsets $X,Y$ of vertices of $G$, both of size $\ell$
and inducing all edges between $X$ and $Y$. We define similarly \emph{empty $\ell$-bipartite graph} when there is no edge between $X$ and $Y$. Observe that we do not require any condition inside $X$ or $Y$. A class of graphs $\mathcal{C}$ has the \emph{strong Erd\H{o}s-Hajnal property}, introduced in~\cite{FoxPach08} if there exists a constant $c>0$ such that every graph of $\mathcal{C}$ contains an empty $cn$-bipartite graph or a complete $cn$-bipartite graph. As we will see later, the strong Erd\H{o}s-Hajnal property implies the Erd\H{o}s-Hajnal property. The remaining of the proof consists in showing that the class of graphs with no hole and no anti-hole of length at least $k$ has the strong Erd\H{o}s-Hajnal property. 

\begin{lemma} \label{sparsedominatingpath}
Let $G$ be a sparse graph with no hole of length at least $k$, that admits a dominating induced path $P$. Then $G$ contains an empty $cn$-bipartite graph. Here $c$ depends of the coefficient $\varepsilon$ of sparsity, and $k$.
\end{lemma}

\begin{proof} 
Let us consider a subpath $I$ of $P$ of length $k$. We assume that $P$ is given in a left right order from one endpoint to the other. The vertices of $G\setminus P$ fall into three categories: the {\it left} of $I$ denotes the vertices with all neighbors
in $P$ at the left of $I$, the {\it right} of $I$ denotes the vertices with all neighbors
in $P$ at the right of $I$, and the {\it inside} of $I$ denotes the other vertices of $G\setminus P$. Observe that if a vertex has both a neighbor at the left and a neighbor at the right of $I$, but no neighbor in $I$, then there is a hole of length at least $k$. Since $P$ is a dominating set, a vertex inside of $I$ that has no neighbor in $I$ must have by definition both a neighbor at the left and a neighbor at the right of $I$, which provides a long hole. It follows that every vertex inside of $I$ has a neighbor in $I$. Similarly, note that there is no edge between the left of $I$ and the right of $I$. So the left of $I$ and the right of $I$ form an empty bipartite graph.\\
We claim that if $G$ is sparse, then the inside vertices cannot be too many. This is straightforward if the degree is bounded by $\varepsilon n$, since the inside vertices belong to the neighborhood of one of the $k$ vertices of $I$. If there is no triangle in $G$, then the neighborhood of every vertex of $I$ is a stable set, hence the neighborhood of $I$ has chromatic number at most $k$. Consequently, if the inside of $I$ is large, then there is a large stable set, and then empty bipartite graph in it. 
Since every sparse graph has maximum degree $\varepsilon n$ or has no triangle, we can assume that for every $I$, the inside of $I$ is bounded in size by some small $\delta n$.
Now take $I$ to be the rightmost $k$-subpath of $P$ that has more right vertices than left ones. Observe that both the left and the right of $I$ contain close to $(\frac{1}{2}-\delta)n$ vertices, hence a large empty bipartite graph.
\end{proof}

A graph on $n$ vertices is an {\it $\varepsilon$-stable set} if it has at most 
$\varepsilon {n\choose 2}$ edges. The complement of an $\varepsilon$-stable set is 
an {\it $\varepsilon$-clique}. Fox and Sudakov proved the following in~\cite{Fox08}. A stronger version of the following result was proved 
by R\"odl~\cite{Rodl86}.

\begin{theorem} \label{reg}
For every positive integer $k$ and every $\varepsilon >0$, there exists $\delta>0$ such that every graph $G$ on $n$ vertices satisfies one of the following:
\begin{itemize}
\item $G$ induces all graphs on $k$ vertices.
\item $G$ contains an $\varepsilon$-stable set of size at least $\delta n$.
\item $G$ contains an $\varepsilon$-clique of size at least $\delta n$.
\end{itemize}
\end{theorem}

The proof of R\"odl is based on the Szemer\'edi's regularity lemma. The proof of Fox and Sudakov provides a much better estimate with $\delta = 2^{-ck (\log 1/\epsilon)^2}$ with a rather different method.
We now prove our main result: 

\begin{theorem} \label{path}
For every $k$, the class of graphs with no hole nor anti-hole of size at least $k$ has the strong Erd\H{o}s-Hajnal property.
\end{theorem}

\begin{proof} Let us first prove that we can restrict the problem to sparse connected graphs without long holes. 
Indeed, since $G$ contains no hole of size $k$, it does not induce all graphs on $k$ vertices, and Theorem~\ref{reg} ensures that $G$ contains an $\epsilon$-clique or an $\epsilon$-stable set of linear size. If $G$ contains an $\epsilon$-stable set $X$, then we delete all the vertices of degree at least $2\epsilon |X|$. Since the average degree is at most $\epsilon$, at most one half of the vertices are deleted. The remaining vertices have maximum degree at most $2\epsilon$, which provides a $4\epsilon$-sparse graph.
If $G$ contains an $\epsilon$-clique of linear size, then $\overline{G}$, which also satisfies the theorem hypotheses, contains a linear-size $\epsilon$-stable set. Thus $\overline{G}$ contains an empty or complete linear-size bipartite graph, and symmetrically, so does $G$.
Finally, we can assume that $G$ is connected: it suffices to apply the theorem on a large connected component if any, or to assemble the connected components in order to get a large empty bipartite graph.\\
Let $G$ be a connected sparse graph with no long hole. We consider the tree $T$ resulting from our algorithm with an arbitrary root. To every node $v$ of $T$ we associate a weight equal to the number of vertices $x$ of $G$ with $r(x)=v$. Note that the total weight equals $n$. By Lemma~\ref{weighttree}, we find in $T$ a path with weight at least $\frac{n}{4}$ or two unrelated subsets of size at least $\frac{n}{4}$. In the first case, the graph $G$ contains a subgraph of size $\frac{n}{4}$ which is dominated by an induced path, and we conclude using Lemma \ref{sparsedominatingpath}. The second case yields an empty $\frac{n}{4}$-bipartite graph, as Lemma \ref{keylemma} ensures that there is no edge between vertices in two unrelated sets.
\end{proof}

Finally, we can prove Theorem~\ref{thm:ehhole} using the following classical result due to Alon et al.~\cite{Alon05} and Fox and Pach~\cite{FoxPach08}.

\begin{theorem}[\cite{Alon05, FoxPach08}] \label{bipartite}
If $\mathcal C$ is a class of graphs that admits the strong Erd\H{o}s-Hajnal property, then 
$\mathcal C$ has the Erd\H{o}s-Hajnal property.
\end{theorem}

\begin{proof}[Sketch of the proof]
Let $c>0$. Assume that every graph of the class $\mathcal{C}$ has a complete $cn$-bipartite graph or an empty $cn$-bipartite graph. Let $c'>0$ such that $c^{c'}\geq \frac{1}{2}$. We prove by induction that every graph $G$ of $\mathcal C$ induces a $P_4$-free graph of size $n^{c'}$. By our hypothesis on $\mathcal C$, there exists, say, a complete $c \cdot n$-bipartite graph $X,Y$ in $G$. By applying the induction hypothesis independently on $X$ and $Y$, we form a $P_4$-free graph on $2(c \cdot n)^{c'}\geq n^{c'}$ vertices. The Erd\H{o}s-Hajnal property of $\mathcal C$ follows from the fact that every $P_4$-free $n^{c'}$-graph has a clique or a stable set of size at least $n^{\frac{c'}{2}}$.
\end{proof}
Combining Theorem~\ref{path} and~\ref{bipartite} ensures that graphs with no long hole nor anti-hole satisfy the Erd\H{o}s-Hajnal conjecture.

\bibliographystyle{plain}

\end{document}